\documentclass[runningheads]{llncs}
\usepackage[T1]{fontenc}
\usepackage{graphicx}
\usepackage{amsmath}
\usepackage{amsfonts}
\usepackage{amssymb}
\usepackage{mathabx}
\usepackage{bussproofs}
\usepackage{lineno}
\usepackage{enumerate}
\usepackage{enumitem}

\usepackage{stmaryrd}
\usepackage{euscript} 
\usepackage{latexsym}
\usepackage{mathrsfs}
\usepackage[all]{xy}
\usepackage{dsfont}

\usepackage{url}
\usepackage{enumitem}
\usepackage{xcolor}
\usepackage{colortbl}
\usepackage{colonequals}
\usepackage[all]{xy}
\usepackage{tikz}
\usetikzlibrary{positioning,arrows,calc}

\usepackage[sort]{cite}
\usepackage{bookmark}

\usepackage[utf8]{inputenc} 

\usepackage{booktabs}       
\usepackage{amsfonts}       
\usepackage{amsmath}
\usepackage{amssymb}
\usepackage{mathrsfs}
\usepackage{tabularx}
\usepackage{tikz-cd}
\usepackage{tikz}
\usepackage{caption}
\usetikzlibrary{matrix}
\usepackage{nicefrac}     
\usepackage{multirow} 
\usepackage{float}
\usepackage{microtype}      
\usepackage{doi}
\usepackage[nottoc]{tocbibind}
\usepackage{chngcntr}
\usepackage{enumerate}
\usepackage{blindtext}
\usepackage{stmaryrd }
\usepackage{nomencl}
\usepackage{etoolbox}
\usepackage[draft]{fixme}

\usepackage{mathrsfs}
\usepackage{tikz,pgfheaps}
\usetikzlibrary{shapes,arrows,positioning,decorations,calc,snakes}
\usepackage{ctable}
\usepackage{etex}

\usepackage{tabularx}
\usepackage{booktabs}
\usepackage{caption}
\usepackage{graphicx}
\usepackage{tikz-cd}
\usepackage{xcolor}
\usepackage{tikz}
\usetikzlibrary{positioning,arrows,calc}
\usepackage{bussproofs}
\usepackage{amsmath}
\usepackage{amssymb}
\usepackage{mathrsfs}
\usepackage{tabularx}

\usepackage{caption}
\usetikzlibrary{matrix}
\usepackage{mathrsfs,proof}
\usetikzlibrary{shapes,arrows,positioning,decorations,calc,snakes}
\usepackage{ctable}
\usepackage{bussproofs,
}

\usetikzlibrary{positioning,arrows,calc}

\newcommand\ndel{{\mathop{\neg\triangle}}}
\newcommand\del{{\mathop{\triangle}}}
\newcommand\imm{\ensuremath}
\newcommand\limp{\to}

\newcommand\axdi{{\ensuremath{\axdelta1}}}
\newcommand\axdii{{\ensuremath{\axdelta2}}}
\newcommand\axdiii{{\ensuremath{\axdelta3}}}
\newcommand\axdiv{{\ensuremath{\axdelta4}}}
\newcommand\axdv{{\ensuremath{\axdelta5}}}

\newcommand\axdr{{\ensuremath{\axdelta6}}}

\newcommand\Proves\Rightarrow
\newcommand\nProves\nRightarrow

\newcommand\I{\imm{\mathcal I}}

\newcommand\lOr{\bigvee}
\newcommand\lequiv{\leftrightarrow}

\newcommand\axdelta{\mathnormal{\mathrm{\Delta}}}
\newcommand\axdvi{{\ensuremath{\axdelta6}}}

\newcommand\proves{\mathrel{\vdash}}

\usepackage{xcolor}

\newcommand{\Ch}{\mathcal{C}}

\begin{document}
\title{Gödel Logics: On the Elimination of The Absoluteness Operator}
%
%
\author{Matthias Baaz\inst{1}\orcidID{0000-0002-7815-2501}\thanks{This research was funded in part by the Austrian Science Fund (FWF) 10.55776/P36571.} \and
Mariami Gamsakhurdia\inst{1}\orcidID{0009-0009-0468-9945} }
\authorrunning{M. Baaz, M. Gamsakhurdia}
%
\institute{Institute of Discrete Mathematics and Geometry, TU Wien, Austria \\\email{\{baaz,mariami\}@logic.at}
}
\maketitle              
\begin{abstract}
 We investigate the eliminability of the absoluteness operator $\del$ in Gödel logics. While $\del$ is not definable from the standard connectives and disrupts important proof-theoretic properties, we show that it becomes eliminable at the propositional level under a non-standard semantics in which all propositional atoms (except the truth constant $\top$) are interpreted strictly below 1.
Under this semantics, every formula containing $\del$ is equivalent to a disjunction of chain formulas, yielding a $\del$-free normal form (standard and non-standard semantics coincide w.r.t. valid formulas without $\del$). 
We further analyze the situation in the first-order setting, where $\del$-elimination fails in general due to recursion-theoretic and topological constraints, but can be recovered under witnessed semantics.

\end{abstract}

\section{Introduction}
Gödel logics form a central class of intermediate logics, characterized by truth values in linearly ordered subsets of the unit interval $[0,1]$. While they enjoy elegant semantic and proof-theoretic properties, they exhibit a fundamental structural asymmetry between the truth values $0$ and $1$.
In standard Gödel semantics, the value $0$ is distinguishable via negation, whereas the value $1$ is not. This asymmetry arises from the continuity of all logical connectives at $1$: while $\neg \neg A$ detects whether $A>0$, there is no connective that detects whether $A< 1$. Consequently, the top truth value behaves in a fundamentally different way from the bottom one.
To overcome this limitation, the absoluteness operator $\del$ has been introduced by Baaz in \cite{Baaz:Delta}. Semantically, $\del A$ evaluates to $1$ if and only if $A$ evaluates exactly to $1$, and to $0$ otherwise. 

However, this comes at a cost. The $\del$-operator is not definable from the standard connectives and fundamentally alters the logical landscape.
In particular, adding $\del$ affects key proof-theoretic properties, such as the deduction theorem and substitution closure\footnote{Note that the equivalence principle is not valid when $\del$ is added.}, and increases the complexity of the logic. 
This leads to a central problem:
\begin{itemize}[nosep]
    \item Semantically, $\del$ is natural and expressive, as it captures the notion of absolute truth and distinguishes between true and false.
    \item Proof-theoretically, $\del$ is problematic: it complicates cut-elimination, and resists elimination in standard semantics as shown in \cite{Baaz-Preining-Zach:CompletenessGoedelDelta}.
\end{itemize}

In this paper, we will analyze under which circumstances $\del$ becomes eliminable.
The elimination of $\del$ is motivated by several factors: it simplifies the investigation of prenex fragments \cite{Gamsakhurdia:Prenex}, enhancing semantic interpretations, axiomatizations and decision procedures\cite{Gamsakhurdia:ismvl}; it clarifies the logical relevance of $\del$, explicitly showing when it impacts formula-validity; restores substitution–closure (lost once \(\del\) is added); and it reveals how $\del$ influences logic completeness, especially at the first-order level. In addition, we show that standard and non-standard semantics coincide at the propositional level without $\del$ (Proposition \ref{equivalence}). 

In this paper we distinguish two notions of $\del$-elimination: one
 under which the effective validity-equivalence is preserved, since in the propositional case, validity is decidable and thus validity-equivalence is always effective \cite{BaazFasching2009}, and in a much stronger sense when the logical-equivalence is preserved. Effective validity-equivalence elimination holds in both standard and non-standard semantics. This result shows that $\del$ does not increase expressive power at the level of effective validity. However, it does not preserve logical equivalence. The logical-equivalence elimination holds in non-standard semantics but fails in standard semantics. The difference between these notions is trivial only in the first-order setting. For that purpose, we introduce a novel non-standard semantics, which, to our knowledge, has not been previously studied in the literature. 

The first half of the paper studies $\del$-elimination under non-standard semantics at the propositional level by transforming formulas into $\del$-free chain normal forms (Corollary \ref{l - 6.1}).
While $\del$-elimination is achievable in the propositional case under non-standard semantics, the situation changes dramatically in the first-order setting. In particular, due to recursion-theoretic and topological constraints on truth value sets, $\del$ cannot, in general, be eliminated effectively. This highlights a sharp boundary between propositional and first-order Gödel logics. In the second half of the paper, we extend our results to the first-order setting, where, in addition to the non-standard semantics, we also assume witnessed interpretations.

A closely related distinction arises between \emph{standard} and \emph{witnessed} Gödel logics. In standard semantics, quantifiers are interpreted via infima and suprema, which need not be attained. In witnessed Gödel logics, these extrema are required to be realized by actual elements of the domain.
This difference has significant consequences. In witnessed logics, many classical principles, such as quantifier shift rules, become valid, and, importantly, the $\del$-operator can be eliminated syntactically via suitable translations\cite{BaazFasching2009} preserving the effective validity-equivalence. By contrast, in non-witnessed settings, $\del$ remains intrinsically tied to the semantics.
By combining witnessed semantics with non-standard semantics we obtain a $\del$-elimination mechanism which preserves logical equivalence, and therefore is much stronger result.


\section{Preliminaries}

Gödel logics $G_V$, are a family of many-valued logics (where the set of truth values $V$ is a closed subset of $[0,1]$ containing both 0 and 1) that form an essential class of intermediate logics, those that are stronger than intuitionistic logic, yet weaker than classical logic. The language is standard (propositional, first-order) with countably infinite propositional variables $A_i$, connectives~$\land$, $\lor$, $\limp$, and the constants ~$\bot$ for ``false'' and $\top$ for ``true''; Atomic formulas include propositional variables and truth constants. Negation is introduced as an abbreviation: we let
$\lnot  A \equiv (A\to \bot)$.  For convenience, we also define $\top \equiv \bot \limp \bot$.

Throughout this paper we mainly focus on infinite-valued Gödel logics.\footnote{Note that this does not make difference at the propositional level but does on the first-order.} In this section we recall the basic definitions of Gödel logics, 
including their semantics, validity, and the $\del$-operator.
\begin{definition}
G\"odel logics are a family of many-valued logics where the truth values set (known also as \emph{G\"odel set}) \(V\) is closed subset of the full \([0,1]\) interval that includes both $0$ and $1$ given by the following valuation~$\I$ based on~$V$ as a function
  from the set of propositional variables into $V$ given as follows: 
\begin{align*}
   \mathcal{I}(\bot)= 0 \,\,\,\,\,\quad & \mathcal{I}(\top)= 1\\
   \mathcal{I}(A \wedge B)= \min \{\mathcal{I}(A), \mathcal{I}(B)\}   \quad &
    \mathcal{I}(A \vee B)= \max \{\mathcal{I}(A), \mathcal{I}(B)\} \\
     \mathcal{I}(\forall x A(x))= \inf\{\mathcal{I}( A(u))\mid u\in U_\mathcal{I}\} \quad &
 \mathcal{I}(\exists x A(x))= \sup \{\mathcal{I}( A(u))\mid u\in U_\mathcal{I}\}\\
\I(\neg A) = \begin{cases} 0 & \text{if $\I(A)>0$}\\
    1 & \text{otherwise}\end{cases} \quad &  \mathcal{I}(A \rightarrow B) = \begin{cases}\I(B) & \text{if\/ $\I(A) > \I(B)$,} \\
               1 & \text{if\/ $\I(A) \le \I(B)$.}\end{cases} 
 \end{align*}

\end{definition}

\begin{remark}\label{remark1}
We use the following abbreviation throughout the paper:  $A\lequiv B$ for $ (A\rightarrow B) \land (B\rightarrow A)$ and  $A<B$ for \((B\rightarrow A)\rightarrow B\), which evaluates to 1 iff the value of \(A\) is smaller than the value of \(B\) with the exception that both take the value $1$. In particular, \(A<\top\) is always true. 
\end{remark}

\begin{definition} \label{d - 1.4}
A formula in Gödel logic is \emph{valid} 
if the formula evaluates to 1 
under every interpretation. 
\end{definition}

\begin{remark}
  The Gödel logic \(G_V\)  is defined as the set of valid formulas.  Validity and unsatisfiability are not dual in Gödel logic, e.g. $A\vee \neg A$ is not valid but its negation is unsatisfiable. Note that the addition of~ $\del$ might
 actually restore the duality between satisfiability and validity in fragments where $\del$ occurs as outermost connective for every formula, or whenever $\del$ occurs in every atom in the formula.
\end{remark}
Note that in  G\"{o}del logics the validity and unsatisfiability of a formula depends only on the relative ordering and the topological structure of the truth values of atomic formulas, and not on their specific values. 
\begin{definition}\label{d - 1.6} 
A \emph{non-isolated point} of a topological space is a point \(x\) such that for every open neighbourhood \(U\) of \(x\) there exists a point \(y\in U\) with \(y\neq x\). A \emph{limit point} of a topological space is a point \(x\) that is not isolated.  A space is \emph{perfect} if all its points are limit points.
A set \(P\subseteq R\) is \emph{perfect} if it is closed and together with the topology induced from \(R\) is a perfect space. 
\end{definition}
\begin{proposition}[cf.\cite{Baaz-Preining:Gödel–Dummett}]
Whenever $V_1\subseteq V_2$ then $G_{V_2} \subseteq G_{V_1}$, as a consequence $\models_{G_{V_1}}A$ then $\models_{G_{V_2}}A$.
\end{proposition}
 
\begin{proof}
Any counterexample in a smaller Gödel logic can be expressed in bigger logic only using the valuations of variables and truth values.    
Let $V_1\subseteq V_2$. Any interpretation into $V_1$ can be viewed as an interpretation into $V_2$. 
Thus, any countermodel in $G_{V_1}$ is also a countermodel in $G_{V_2}$, yielding the claim.
\end{proof}

\subsection{$\del$-Operator}
For G\"{o}del logics there is an
asymmetry between~$0$ and~$1$, this is because~$0$ can be distinguished from other values (by using the negation), while~$1$ cannot be
distinguished. The reason is that all connectives and quantifiers are
continuous at~$1$. To overcome this asymmetry the absoluteness operator~$\del$ (dual of negation) has
been introduced in~\cite{Baaz:Delta} \begin{align*}    
 \I(\del A) = \begin{cases} 1 & \text{if $\I(A)=1$,}\\ 0 &\text{otherwise}\end{cases}
\end{align*}





\begin{proposition}
$\del$ is not definable using other connectives \(\land,\lor,\to,\bot,\top\) and variables. 
\end{proposition}

\begin{proof}
There is a finite number of 1-variable functions in infinite-valued propositional Gödel logic: $\top ,
\bot, 
A, 
\neg A,
\neg A\lor A,
\neg A\rightarrow A.  $
Assume that $\del$ is definable by some function $f$, i.e., $\del A\lequiv f(A)$.
Now we look at the $f$ in $G_3$, because if $\del$ is not definable in $G_3$ then it is not definable in all larger propositional Gödel logics.
The following truth table shows that none of them defines $\del$ and they are closed under composition by all connectives:  
 \begin{center}
 \begin{tabular}{c|c|c|c|c|c|c}
$A$ & $\neg A$ & $\top$ & $\bot$ & $A\lor \neg A$& $\neg A\limp A$ & $\del A$\\
\hline
0 & 1 & 1 & 0 & 1 & 0 & 0\\
1/2 & 0 & 1& 0&1/2 & 1 & 0\\
1 & 0&1 &0 & 1 & 1 & 1
\end{tabular}   
\end{center}
Therefore we introduce the connective $\del$ extending 
the language. 
\end{proof}
We consider Gödel logics with or without an absoluteness  operator~$\del$ and denote by $G_V^\del$ (sometimes $G_\del$) or $G_V$, respectively.
\begin{theorem}
If a formula with $n$-variables is valid in a propositional logic $G^\del_{n+2}$, then it holds in all propositional Gödel logics with $\del$. Consequently, the infinite-valued propositional Gödel logic with $\del$ is the intersection of the finite ones (not considering the entailment). There is only one infinite-valued propositional Gödel logic with $\del$.
\end{theorem}
\begin{proof}
A propositional formula with $\leq n$-variables is valid in any infinite-valued Gödel logic (with or without $\del$) iff it is also valid in $n+2$ valued propositional Gödel logic (with or without $\del$) since a counterexample can be transformed into a counterexample only using the valuations of variables and truth values. Therefore, the infinite-valued Gödel logic with $\del$ is the intersection of the finite ones, hence there is only one infinite-valued propositional Gödel logic with $\del$. Note that this is not the case in first-order setting \cite{Baaz-Preining:Gödel–Dummett}. 
\end{proof}

$\del$ allows the definition of a strict order $A<_\del B\equiv \ndel(B\rightarrow A)$.
Baaz gave the following axiomatization of Gödel logics with $\del$ using Hilbert style calculus \cite{Baaz:Delta}:
  \begin{align*}
    \axdi &\quad \del A\lor \lnot\del A \quad &
    \axdii &\quad \del (A\lor B) \limp (\del A \lor \del B)\\
    \axdiii&\quad \del A \limp A \quad &
    \axdiv &\quad \del A \limp \del\del A\\
    \axdv   &\quad \del (A\limp B)\limp(\del A\limp\del B) \quad &
    \axdr & \quad A \proves \del A
  \end{align*}
  
Then $G^{\del}_{[0,1]}$ is axiomatized by $G_{[0,1]} + \del.$

\begin{remark} 
Propositional Gödel logics can be identified through well-founded linear Kripke structures. 
The $\del$-operator in Gödel logics can be interpreted as a stability operator, meaning:
 $\del A$ holds if and only if $A$ is true in all future and past worlds. 
$\del A$ being true in given world implies that $A$ has also been true in all past worlds, \(\del\) semantically coincides with modal logic $S4$ with additional axioms $\axdi, \axdii $\cite{Baaz:Delta}.


\end{remark}

\begin{remark}

   \begin{itemize}[nosep]
    \item At the propositional level all infinite-valued Gödel logics
    admit hypersequent calculi \cite{Baaz-Preining-Zach:CompletenessGoedelDelta}; for the first-order case 
    only Gödel logic corresponding to the full interval \([0,1]\) admits hypersequent calculus. In case of uncountable truth set with $0$ isolated, the existence of such calculi is unknown, in all other case, infinite-valued Gödel logics are not r.e.  \cite{Baaz-Ciabattoni:TaitCutElimGoedel}.
      
    \item Standard Gödel logics enjoy the full deduction theorem; adding \(\del\) breaks this property in general.
    \item Propositional Gödel logics with $\del$  cannot be considered as intermediate logics \cite{Baaz-Preining:Gödel–Dummett}.    

 \end{itemize}
\end{remark}

\subsection{Non-standard Semantics}\label{sec:rest}

In general, the absoluteness operator $\del$ is not eliminable in Gödel logics. However, with a slight modification of the semantics, we show that $\del$ can be eliminated at the propositional level. Notably, while standard Gödel logics are closed under substitutions, this closure no longer holds in the non-standard semantics where $\del$ is eliminable.

\begin{definition}\label{def:rest}
   Let $V \subseteq [0,1]$ be a Gödel set. 
An interpretation $\I$ satisfies the non-standard semantics if $\I(A_i) < 1$
for every propositional variable $A_i$ (and more generally, for every atomic formula), 
except for the truth constant $\top$.
We denote the corresponding logic by $G_V^*$ and corresponding non-standard interpretation by $\I^*$.
  \end{definition}
Intuitively, this constraint prevents atomic formulas from attaining the maximal truth value. 
As a result, the operator $\del$, which detects exact truth, loses its ability 
to distinguish atomic formulas from compound ones, enabling its elimination.

One of the key observations of this paper is that non-standard semantics can be captured syntactically by adding simple axioms to standard Gödel logic.
Non-standard semantics requires that all atomic formulas (except $\top$) are interpreted strictly below $1$. This semantic condition can be enforced by axioms excluding the possibility that atoms take the value $1$. Thus, the non-standard semantics is axiomatized by extending Gödel logic with $\neg \del X_i$
for all atoms $X_i$ except $\top$. 
In the propositional settings, for each propositional variable $A_i$ occurring in a formula, we add the axiom $\neg \del A_i$. This expresses that $A_i$ never evaluates to $1$.
In the first-order case, the conjunction ranges over all predicate symbols occurring in a formula, using universally quantified instances $\forall x\, \neg \del P(x)$. This ensures that no instance of $P$ attains the value $1$.

Importantly, these axioms are required only for the atomic formulas that actually occur in the given formula. Thus, the axiomatization is \emph{local} and depends on the formula under consideration.

\begin{proposition}(Reduction to Standard Semantics).
 Let $F(X_1,\dots,X_n)$ be a propositional formula. Then the following are equivalent:

\begin{enumerate}[nosep]
\item $F$ is valid under non-standard semantics;
\item The formula $(\neg \del X_1 \land \cdots \land \neg \del X_n) \to F$ is valid in standard Gödel logic with $\del$.
\end{enumerate}
\end{proposition}

\begin{proof}
($\Rightarrow$) Assume $F$ is valid under non-standard semantics. 
Let $\I$ be any interpretation in standard semantics. 
If $\I(X_i) < 1$ for all $i$, then $\del X_i = 0$, hence $\neg \del X_i = 1$, 
and the antecedent in 2. evaluates to $1$. Thus $\I^*(F) = 1$.

($\Leftarrow$) Assume the implication is valid. 
Let $\I$ be a non-standard interpretation, so $\I^*(X_i) < 1$ for all $i$. 
Then $\del X_i = 0$, hence $\neg \del X_i = 1$, 
and therefore the implication forces $\I^*(F) = 1$. 
Thus $F$ is valid in the non-standard semantics.
\end{proof}

We have the following immediate consequence:
\begin{proposition}
 In the standard semantics $(\neg \del X_1\land \dots \land \neg \del X_n) \limp F$  is validity-equivalent to  $F\lor  X_1\lor \dots \lor X_n$.   
\end{proposition}

\begin{proof}
Note that the axiom $\axdi$ holds also in the non-standard semantics, therefore we have $A\limp B \equiv \neg A \lor B$.
Using the fact that $\neg \neg \del X_i \equiv \del X_i$ and by $\axdiv$, 
we have $\del X_i\limp X_i$, thus
\[(\neg \triangle X_1 \land \cdots \land \neg \triangle X_n) \to F
\equiv F \lor \triangle X_1 \lor \cdots \lor \triangle X_n \equiv F \lor X_1 \lor \cdots \lor X_n.\]
In the other direction, by $\axdr$, $F \lor X_1 \lor \cdots \lor X_n$ derives $\del (F \lor X_1 \lor \cdots \lor X_n)\equiv \del F \lor \del X_1 \lor \cdots \lor \del X_n$, therefore, by the axiom $\axdiv$ again, it is equivalent to $(\neg \del X_1\land \dots \land \neg \del X_n) \limp F$.
\end{proof}

\begin{remark}
    As a consequence, decision procedures for standard propositional Gödel logics with $\del$ can be applied to non-standard semantics via the above reduction.
    The recursive enumerability status of non-standard Gödel logics is implied by the corresponding standard Gödel logics. 
    non-standard semantics can be fully captured by adding predefined number of axioms to the standard system for each $\del-$formula.
    From the viewpoint of computability and axiomatization, non-standard semantics behaves like a definitional extension rather than a genuinely new logic.
\end{remark}

\subsection{Witnessed Semantics}

In standard first-order Gödel logic, quantifiers are interpreted using infimum and supremum:
\[
\mathcal{I}(\forall x\, A(x)) = \inf_{u \in D} \mathcal{I}(A(u)), 
\qquad
\mathcal{I}(\exists x\, A(x)) = \sup_{u \in D} \mathcal{I}(A(u)).
\]

\begin{example}
Suppose the truth values of the instances of $A(x)$ are
$0.6, 0.55, 0.52, 0.51, \\0.505, \dots$. Then
$\inf \{ \mathcal{I}(A(u)) : u \in |\mathcal{I}| \} = 0.5.$
However, no element of the domain actually yields the value $0.5$.  
Thus the infimum is not attained by any instance of the formula.
\end{example}
Therefore we introduce witnessed semantics:

\begin{definition}
In \emph{witnessed Gödel logics}, the above definition is strengthened by requiring that these extrema are realized by elements of the domain:
\[
\mathcal{I}(\forall x\, A(x)) = \min_{u \in D} \mathcal{I}(A(u)), 
\qquad
\mathcal{I}(\exists x\, A(x)) = \max_{u \in D} \mathcal{I}(A(u)).
\]
\end{definition}

This distinction is crucial for the interaction with $\del$. Note that we exclude interpretations where proper infimum and minimum occur.



Witnessed Gödel logic can be axiomatized by extending Gödel logic with schemas ensuring the existence of witnesses \cite{BaazFasching2009}:
 \[\text{Infimum is attained:} \quad 
    \exists x\, (A(x) \rightarrow \forall y\, A(y)),
    \]
\[\text{Supremum is attained:}\quad 
    \exists x\, (\exists y\, A(y) \rightarrow A(x)).
    \]


An important consequence of the witness axioms is that certain
classically valid quantifier shift rules become valid in witnessed
Gödel logics.
In particular, the following formulas, which are not
valid in ordinary Gödel logic, become valid in witnessed Gödel logic:
\[
(\forall x\, A(x) \rightarrow B) \rightarrow \exists x\,(A(x) \rightarrow B)\quad 
(A \rightarrow \exists x\, B(x)) \rightarrow \exists x\,(A \rightarrow B(x)).
\]

In ordinary Gödel logic, the value of $\forall x\, A(x)$ may arise from a non-attained infimum, so there need not exist an element $x$ that realizes this value. Hence, the quantifier shift may fail. In witnessed Gödel logic, however, such an element always exists. Therefore, the classical quantifier shift becomes valid. This is one of the main conceptual advantages of witnessed semantics.
 We denote the use of witnessed semantics by $GW_V$.


We now analyze the interaction between quantifiers and the $\del$-operator.
The universal quantifier behaves well with respect to $\del$:
\[
\del \forall x\, A(x) \;\leftrightarrow\; \forall x\, \del A(x) \footnote{
It is shown that all Gödel logics can be identified with linear Kripke frames of constant-domain \cite{Beckmann-Preining:LinearKripkeGoedel}.}.
\]
In contrast, the existential quantifier exhibits an asymmetry:
\[
\exists x\, \del A(x) \;\rightarrow\; \del \exists x\, A(x),
\]
but the converse implication fails in general\footnote{ $\exists x\, \del A(x)$ requires an element whose value is exactly $1$, $\del \exists x\, A(x)$ only requires that the supremum is $1$.}.
The asymmetry between $\exists$ and $\del$ can be understood semantically.
In standard Gödel logic: $\mathcal{I}(\exists x\, A(x))$ is defined as a supremum, the supremum may be $1$ without any element attaining the value $1$.
Thus, $\mathcal{I}(\del \exists x\, A(x)) = 1$
may hold even though for all $u$ we have $ \mathcal{I}(A(u)) < 1.$
In this situation, $\mathcal{I}(\exists x\, \del A(x)) = 0.$
In witnessed Gödel logic, this pathology disappears: if $\mathcal{I}(\exists x\, A(x))$ has value $1$, then there exists $u$ such that $\mathcal{I}(A(u)) = 1$, hence $\exists x\, \del A(x)$ holds as well.
Thus, witnessed semantics restores a closer alignment between $\exists$ and $\del$.
This improved interaction is one of the main reasons why $\del$ can be eliminated syntactically in witnessed Gödel logics, in contrast to the standard first-order setting.

To obtain witnessed Gödel logic with $\del$, we add the propositional axioms for $\del$ given on page 5, and interaction principles between $\del$ and quantifiers. Then $GW^{\del}_{[0,1]}$ is axiomatized by
\[
GW_{[0,1]} + \del + \exists x\,\del A(x) \leftrightarrow \del \exists x\,A(x).
\]



\section{ Propositional Case}\label{propositional}

This section contains the core technical results of the paper. We analyze the behavior of the $\del$-operator at the propositional level and establish conditions under which it can be eliminated.

In standard Gödel logic, the so-called equivalence principle is stable under arbitrary contexts: if $A \leftrightarrow B$, then $E(A) \leftrightarrow E(B)$ for any context $E$. This is not the case when $\del$ is present in the language.


\begin{example}
The specific case
 $(A\lequiv B) \rightarrow (\del A\lequiv \del B)$  for the $\del$ operator also fails in the non-standard semantics.
 To illustrate, assign to $A$ value $1$ and to $B$ some value strictly between 0 and 1. In this case, $\del A$ is $1$ and $\del B$ is 0, yet $A\lequiv B$ is not 0. This contradiction demonstrates why the principle does not hold universally.
\end{example}

\begin{remark}
Unlike standard Gödel semantics, non-standard semantics is not
closed under arbitrary uniform substitution. Indeed, the
constraint imposed on atomic valuations need not be preserved
when propositional variables are replaced by arbitrary formulas.
For example, suppose that every propositional variable is
required to take a truth value strictly below $1$. Then
\[
G_v^*\models \neg\Delta p,
\]
since $\I^*(p)<1$ implies $\I^*(\Delta p)=0$ for every non-standard
valuation $\I^*$. However, substituting $\top$ for $p$ yields
\[
G_v^*\not\models \neg\Delta\top,
\]
because $\I^*(\Delta\top)=1$.
Nevertheless, non-standard validity is preserved under substitutions
$\sigma$ satisfying
\[
\I^*(\sigma(p))<1
\]
for every non-standard valuation $\I^*$ and every propositional
variable $p$ in the domain of $\sigma$. Indeed, such substitutions
induce non-standard valuations, so the usual substitution lemma
applies. In particular, non-standard validity is preserved under
renaming of propositional variables.

Notice that this failure of uniform substitution does not imply
a failure of replacement of semantically equivalent formulas.
If $\I^*(A)=\I^*(B)$ for every non-standard valuation $\I^*$, then $A$
and $B$ may be interchanged in any formula without changing
its truth value.
\end{remark}
Therefore, we must modify the evaluation process for $\del$ to accommodate this limitation.
First, we show that
\begin{lemma}\label{lemma1} 
For any formulas $A,B,C,D$, and any context function $E$, the following holds:
   from $(A\lor B) \land (C\lequiv D)$ we derive $ (A\lor B) \land (E(C)\lequiv E(D))$ 
 
\end{lemma}
\begin{proof}
   The result follows by induction on the complexity of the context function $E$.
   This structure extends naturally to the transformed formula.
   Note that $A$ derives $\del A$, $\del$ distributes over connectives and $\del A$ derives $A$.
   
    \begin{center}
\AxiomC{$ (A\lor B) \land (C\lequiv D)
$}
\UnaryInfC{$\del( (A\lor B) \land (C\lequiv D))$}
\UnaryInfC{$\del A\lor \del B \land \del(C\lequiv D)
)$}
\UnaryInfC{$\del A\limp A$ \quad $\del B \limp B$ \quad $\del A\lor \del B \land (\del C\lequiv \del D)$}
\UnaryInfC{$ A\lor  B \land (\del C\lequiv \del D)$}
\DisplayProof
\end{center}
and in consequence for any context function $E$.
\end{proof}

\begin{remark}
This lemma holds for both standard and non-standard semantics. 
This is used to push implication-clauses into contexts, needed to eliminate inner \(\del\) occurrences. 
\end{remark}

\begin{lemma}\label{lemma2}
The transitive closure of equivalence in the non-standard semantics holds, i.e., for any formulas $C,D,L$:
\begin{center}
    $ (A\lor B) \land ((C\lequiv D) \land (D\lequiv L)) \Rightarrow (A\lor B) \land (C\lequiv L)$.
\end{center}

\end{lemma}
\begin{proof}
   This follows from the fact that in general in all Gödel logics  \begin{center}
\AxiomC{$ C\lequiv D
$}
\AxiomC{$  D\lequiv L
$}
\BinaryInfC{$C\lequiv L$}
\DisplayProof
\end{center}
Therefore
     \begin{center}
\AxiomC{$ (A\lor B) \land (C\lequiv D)
$}
\AxiomC{$ (A\lor B) \land (D\lequiv L)
$}
\BinaryInfC{$\ (A\lor B) \land (C\lequiv L)$}
\DisplayProof
\end{center}

\end{proof}

\begin{lemma} \label{lemma3}
  Given any expression $\del(C\lor (D\land a))$ for some variable $a$ in the non-standard semantics where $C, D$ are valid expressions we can eliminate $\del$ obtaining $C$
\begin{center}
    \AxiomC{$\del(C\lor (D\land a))$}
\UnaryInfC{$C$}
\DisplayProof
\end{center}
  
\end{lemma}
\begin{proof}
Recall that $(C\lor (D\land a))$ derives $\del(C\lor (D\land a))$ by $\axdvi$, distribute and eliminate $\del$ considering the fact that $C, D$ are valid expressions, as follows

    \begin{center}
    \AxiomC{$(C\lor (D\land a))$}
    \UnaryInfC{$\del(C\lor (D\land a))$} 
  \UnaryInfC{$\del C \lor \del(D\land a)$} 
   \UnaryInfC{$ C \lor( \del D\land  \del a)$} 
   \UnaryInfC{$ C \lor(  D\land  \bot)$} 
\UnaryInfC{$C\lor\bot$}
\UnaryInfC{$C$}
\DisplayProof
\end{center}

\end{proof}



\subsection{Structural Normal Form: Chains}\label{chains}

The key idea is that Gödel interpretations induce linear preorders on formulas. Chains encode these orderings syntactically, allowing us to reduce semantic reasoning to finitely many order types. A fundamental insight is that chains in Gödel logic play the role of constituents of disjunctive normal form (DNF) in classical logic. The elimination of $\del$ is based on a structural decomposition of valuations into linear orderings.
To understand the elimination of $\del$, it is crucial to analyze its behavior in the context of chain structures. Specifically, we must consider only those chains that do not equate the variable directly or to 1. 

  \begin{definition}\label{d - 6.1}
 A \emph{chain} $C$ 
 over the set of propositional variables $\{X_1,...X_n\}$ is an expression (ordering of these variables relative to the constants \(\top\) and \(\bot\), described by positions of variables)
 \begin{align*}
\bot \vartriangleright_1 X_{\Pi(1)}\land X_{\Pi(1)} \vartriangleright_2
         X_{\Pi(2)}\land \dots  \land X_{\Pi(n)}\vartriangleright_{n+1} \top
        \end{align*}
    where  $\Pi$ is a permutation on $[1\dots n]$ and  $\vartriangleright_i\in \{<, \lequiv\}$, 
   with $X_{\Pi(i)}\vartriangleright_i X_{\Pi(i+1)}$ to be understood as $X_{\Pi(i)}\lequiv X_{\Pi(i+1)} \equiv X_{\Pi(i)}\limp X_{\Pi(i+1)} \land X_{\Pi(i+1)}\limp X_{\Pi(i)}$ and $X_{\Pi(i)} < X_{\Pi(i+1)}$ to be understood as 
$(X_{\Pi(i+1)}\limp X_{\Pi(i)})\limp X_{\Pi(i+1)}$ due to the Remark \ref{remark1}. Furthermore, $\bot$ should not be in the equivalence class of $\top$. 
\end{definition}
\begin{remark}
    Note that the chain is a formula 
    of whose 1-satisfiability might differ in
 in the standard and non-standard semantics.
\end{remark}

\begin{definition}\label{evaluation}
Let $C$ be a chain on variables $\{X_1,...X_n\}$ and $A$ be a formula with variables among $\{X_1,...X_n\}$. We define a syntactic evaluation\footnote{Evaluating a formula on a chain means simplifying it using the known ordering. Therefore, once the chain tells us which input is smaller, we can evaluate without choosing particular numbers.} of a formula on a chain with the following equivalences:\\
\[1) \quad (C\land A(X_i\land X_j)) \equiv (C\land A(X_i)) \,\,\, \text{if} \,\, \,X_i\rightarrow X_j \,\,\, \text{follows from}\,\,\, C,\]
\[
(C\land A(X_i\land X_j)) \equiv (C\land A(X_j)) \,\,\, \text{else.}\]
\[2) \quad (C\land A(X_i\lor X_j) )\equiv (C\land A(X_i) )\,\,\, \text{if} \,\, \,X_j\rightarrow X_i \,\,\, \text{follows from}\,\,\, C,\]
\[(C\land A(X_i\lor X_j)) \equiv (C\land A(X_j)) \,\,\, \text{else.}\]
\[3) \quad (C\land A(X_i\rightarrow X_j)) \equiv (C\land A(\top) )\,\,\, \text{if} \,\, \,X_i\rightarrow X_j \,\,\, \text{follows from}\,\,\, C, \]
\[(C\land A(X_i\rightarrow X_j)) \equiv (C\land A(X_j) )\,\,\, \text{else.}\]
\end{definition}

\begin{remark}
    Note that the syntactic evaluation described above follows from the semantics of G\"odel logics, and syntactic evaluation
on formulas without $\del$ coincides on standard and non-standard semantics.
\end{remark}
The following proposition describes the exact meaning of reduction under a chain.
\begin{proposition}\label{chainevaluation}
    For any chain $C$ on variables $\{X_1,...X_n\}$ and a formula $A$  without $\del$ with variables among $\{X_1,...X_n\}$
    there exists some atom $P_c\in \{X_1,...X_n, \top,\bot\}$ representative such that 
    \[(C\land A) \equiv (C\land P_c).\]
\end{proposition}
    \begin{proof}
        It follows immediately from the above definitions. Each connective returns either one of its inputs or \(\top\); negation is implication into \(\bot\). No connective creates a new intermediate truth value.
    \end{proof}

Now we can construct the chain normal form. 
Let $\Ch$ contain one chain representation for every exact order type of
the variables $\{X_1,...X_n\}$ under consideration. Every valuation has one such type and
makes its corresponding chain true. Therefore
\begin{equation}\label{eq:cover}
 \bigvee_{C\in\Ch}C\equiv\top.
\end{equation}
The chains need not be pairwise disjoint, and a valuation may make several
chain formulas true. Coverage, rather than uniqueness, is what is needed.

\begin{proposition}\label{propchain}
Let $A$ be a formula different from $\bot$ with variables among $\{X_1,...X_n\}$. We have 
$A \equiv \bigvee_k C_k$ for some disjunction of chains $C_k$ on $\{X_1,...X_n\}$. 
\end{proposition}

\begin{proof}
From the above defintions and Proposition \ref{chainevaluation}, we have
\[A \equiv \bigvee_k (C_k\land P_k)\]
 where $C_k$ are all chains and $P_k$ are  the atoms under syntactic valuation.
For a term $C_k\land P_k$, there are three possibilities:\begin{itemize}
    \item If $P_k$ is $\top$ replace the term by $C_k$.
\item If $P_k$ is $\bot$ replace the term by $\bot$ and discard it from the
 disjunction.
\item If $P_k$ is $X_i$, i.e., lies in an intermediate equality block, recall that $\bot \vartriangleright_1 X_{\Pi(1)}\vartriangleright_2
        X_{\Pi(2)}\vartriangleright_3\dots \vartriangleright_n X_{\Pi(n)}\vartriangleright_{n+1} \top$, 
        merge that block
 and every higher block into the $\top$ block.
The resulting chain will be
 denoted by $C_{\uparrow P_k}$.
\end{itemize}

\end{proof}

 \begin{remark}
     $\bot$ is not equivalent to disjunction of any chains in general, as any chain is satisfiable somewhere and therefore it never has chain normal form which corresponds to the disjunctive normal form.
      The chain normal form for $\top$ corresponds to the Tertium non datur\footnote{Note that, Tertium non datur can be written as a chain normal form as $(\bot\leftrightarrow X< \top) \lor (\bot < X \leftrightarrow \top )$ for some variable $X$.}.
      
 \end{remark}
Below we just consider an example of a chain normal form for a formula in two variables without $\del$ in standard semantics.
\begin{example}
Let us take the two-variable formula $ (A\rightarrow B)\rightarrow B$ and construct a chain normal form. For overview, we construct the chain step by step.
We take one variable $A$ and take every position for its evaluation
\[(\bot\lequiv A)< \top \quad
    \bot<A<\top \quad
    \bot<\top\lequiv A\]
 Now we insert the second variable $B$ in all possible positions given above. Now we list all possible positions for $A$ and $B$;  In the first row we insert the chains in the normal form \footnote{Evaluations of the formula $ (A\rightarrow B)\rightarrow B$  are taken for each position of the variable assigned with the $*$ respectively.}. In the second row, we calculate the first chains by applying the rule that $B=1$ everything above $B$ (inclusive) is $1$.
 \begin{align*}
\mathrm{(1)} & \quad ((\bot\lequiv B^*\lequiv A)<\top\land  \textcolor{red}{B}) \lor & \qquad -  \\
\mathrm{(2)} & \quad     ((\bot\lequiv A)<B^*<\top\land \textcolor{red}{B})\lor & \qquad \mathrm{(1')} & \quad ((\bot \lequiv A)<(B \lequiv \top)) \lor \\
\mathrm{(3)} & \quad ((\bot\lequiv A)<(B^*\lequiv\top) \land \textcolor{red}{B})\lor  & \qquad \mathrm{(2')} & \quad ((\bot \lequiv A)<(B \lequiv \top))\lor \\
\mathrm{(4)} & \quad     ((\bot\lequiv B)<A^*<\top\land \textcolor{red}{ \top})\lor  & \qquad \mathrm{(3')} & \quad ((\bot \lequiv B)<(A < \top))\lor\\
\mathrm{(5)} & \quad     (\bot<B< A^*<\top\land \textcolor{red}{\top})\lor  & \qquad \mathrm{(4')} & \quad (\bot <B<(A < \top))\lor\\
\mathrm{(6)} & \quad     (\bot< A\lequiv B^*<\top \land \textcolor{red}{B})\lor & \qquad \mathrm{(5')} & \quad (\bot < A \lequiv B \lequiv \top) \lor\\
\mathrm{(7)} & \quad     (\bot<A< B^*<\top\land \textcolor{red}{B})\lor & \qquad \mathrm{(6')} & \quad (\bot < A < B \lequiv \top) \lor\\     
\mathrm{(8)} & \quad     (\bot<A< B^*\lequiv\top\land \textcolor{red}{B})\lor  & \qquad \mathrm{(7')} & \quad (\bot < A<B \lequiv \top)) \lor\\
\mathrm{(9)} & \quad     ((\bot\lequiv B)< (A^*\lequiv\top)\land \textcolor{red}{\top})\lor  & \qquad \mathrm{(8')} & \quad ((\bot \lequiv B)<(A\lequiv \top)) \lor\\
\mathrm{(10)} & \quad    ( (\bot< B)< (A^*\lequiv\top)\land   \textcolor{red}{ \top})\lor & \qquad \mathrm{(9')} & \quad ((\bot < B)<(A \lequiv \top)) \lor\\
\mathrm{(11)} & \quad     (\bot<(A\lequiv B^*\lequiv\top)\land \textcolor{red}{B}) & \qquad \mathrm{(10')} & \quad (\bot < A\lequiv B \lequiv \top) 
 \end{align*}
 now we write the optimization of the chains $(1')-(10')$ by removing the copies
 \begin{align*}
 \mathrm{(1)} & \quad ( (\bot \lequiv A)<(B \lequiv \top)) \lor  \\
 \mathrm{(2)} & \quad    ((\bot \lequiv B)<(A < \top)) \lor \\
\mathrm{(3)} & \quad     (\bot <B<A < \top)) \lor\\
\mathrm{(4)} & \quad    ( \bot < A \lequiv B \lequiv \top) \lor \\
\mathrm{(5)} & \quad     (\bot < A < B \lequiv \top) )\lor \\     
\mathrm{(6)} & \quad     ((\bot \lequiv B)<(A\lequiv \top) )\lor\\
\mathrm{(7)} & \quad      (\bot < B<A \lequiv \top) 
 \end{align*}
 
\end{example}

\begin{definition}\label{restrchain}
We define a \emph{non-standard chain} $C^*$ as a chain where there is no variable in the equivalence class of $\top$. 
\end{definition}

\begin{definition}\label{resteval}
The syntactic evaluation of non-standard chains on $\{X_1,...X_n\}$ of formula with variables among $\{X_1,...X_n\}$, possibly containing $\del$ is defined as
above with the addition that $(C\land \del (\bot)) \equiv \bot, \,  (C\land \del (X_i) )\equiv \bot, \, (C\land \del (\top)) \equiv \top$
\end{definition}

\begin{proposition}\label{propo9}
   Let $A$ be a formula (possibly with $\del$) with variables among $\{X_1,...X_n\}$, then
$A \equiv \bigvee_k C_k^*$ for some disjunction of non-standard chains on $\{X_1,...X_n\}$.  
\end{proposition}

\begin{proof}
    The proof is similar to the proof of Proposition \ref{propchain} by using Lemmas \ref{lemma1}, \ref{lemma2}, \ref{lemma3}. By the construction above, we obtain that  $A \equiv \bigvee_k (C_k^*\land P_k)$ for all non-standard chains,
where $P_k\in \{X_1,...X_n, \top,\bot\}$.
\begin{itemize}
    \item If $P_k$ is $\top$, we delete $P_k$. 
\item If $P_k$ is $\bot$, we delete $C_k\land P_k$.
\item  If $P_k$ is some variable $X_i$, we delete $C_k\land P_k$.  
\end{itemize}
\end{proof}

\begin{remark}
    The proof of the above proposition needs an equivalence principle with $\del$ described in Lemmas \ref{lemma1}, \ref{lemma2}, \ref{lemma3}, since usual equivalence principle with $\del$ does not hold.
\end{remark}
\begin{proposition}\label{propo}
In standard semantics, the full disjunction of (all possible) chains $\lOr_k C$ on $\{X_1,...X_n\}$ is valid in all Gödel logics. Similarly, the disjunction of non-standard chains $\lOr_k C^*$ on  $\{X_1,...X_n\}$ is valid in all Gödel logics under the non-standard semantics.
\end{proposition}

\begin{proof}
    Let $\{X_1,...X_n\}$ be a finite set of propositional variables. Any valuation \( \I: X \rightarrow [0,1] \) in Gödel logic induces a total preorder on the set of variables, defined by \( X_i \leq X_j \) whenever \( \I(X_i) \leq \I(X_j) \).
Each such total preorder corresponds to a particular chain \( C_k\), determined by a permutation \( \Pi \) and a choice of relations \( \vartriangleright_i \in \{<, \lequiv\} \) reflecting strict or non-strict orderings in the valuation. Since the set of permutations of \( [1,\dots,n] \) and relation choices is finite, the disjunction over all such chains \( \bigvee_k C_k \) covers all possible total preorders, and hence, all possible valuations.
Therefore, for any Gödel model and any valuation \( \I \), at least one chain \( C_i \) evaluates to \(1\). Hence, \( \bigvee_k C_k \) is valid in every Gödel logic.

The same reasoning applies under non-standard semantics: even when variable assignments are non-standard to values strictly less than \(1\), every assignment still induces a total preorder on the variables (possibly omitting the value \(1\)). The corresponding chains \( C_k^* \), which do not include equalities to \( \top \), are constructed accordingly. Again, at least one chain \( C_k^* \) evaluates to \(1\), so \( \bigvee_k C_k^*\) is valid in all Gödel logics under non-standard semantics.
\end{proof}
\begin{remark}
In the finite-valued Gödel logics with $n$ truth values, chains have length $n$. In the case of infinite-valued Gödel logics, one gets chains with the number of variables incremented by two additional positions for truth constants. 
\end{remark}

\begin{remark}
There might be copies of chains, therefore we restrict definition to only one chain of the specific type, this is possible as arbitrary contraction of identical formulas for disjunction holds.
\end{remark}
Here is the main result of our paper.
\begin{theorem}\label{l - 6.1}
In non-standard semantics each formula $A$ with $\del$ over variables $\{X_1,...X_n\}$ is equivalent to some $\del$-free formula over the same variables $\{X_1,...X_n\}$.
\end{theorem}

\begin{proof}
    By reformulating formulas into chain normal forms, we ensure that $\del$ can be systematically removed while preserving the validity equivalence. The final form is a disjunction of chains without $\del$, which evaluates to 1.  It follows immediately from the Propositions \ref{propo} and \ref{propo9}.
\end{proof}

\begin{remark}\label{remark13}
       The above theorem is valid even when there are free variables present in atoms.
\end{remark}
\begin{remark}
    Given a formula $A$ with $\del$ such that all atoms in the formula are negated or double-negated, we can eliminate $\del$.
\end{remark}

\begin{example}[Eliminating $\del$]
Consider a formula 
    $F:= a\lor \del(a\lor \neg a)$.
The corresponding chain decomposition yields three chains in standard semantics:
\begin{align*}
    C^*_1:\,\, (\bot \lequiv a)< \top, \quad
   C^*_2:\,\, (\bot <a)<\top, \quad
   C_3:\,\, (\bot < a\lequiv \top)
\end{align*}
for some variable $a$.
Note that by Definition \ref{restrchain} the last chain does not occur in the non-standard semantics, therefore we have $C_1\lor C_2.$
Now we construct the chain normal form in the non-standard semantics
\[(F\land C_1)\lor (F\land C_2).\]
Note that we evaluate $F$ from innermost chain first.\\
Evaluation of $C_1$ is as follows
\begin{align*}
    a\lor \del (a\lor \top)\\
    a\lor \del \top\\
    a\lor \top\\
    \top
\end{align*}
Evaluation of $C_2$ is as follows
\begin{align*}
    a\lor \del (a\lor \bot)\\
    a\lor \del a\\
    a\lor \bot\\
   a\\
  \bot
\end{align*}
Thus, we obtain     
        $F \equiv (\bot \lequiv a)$ \text{ or written shorter } $F \equiv \neg a$.
\end{example}

    \begin{remark}
        By reformulating formulas into chain normal forms, we ensure that $\del$ can be systematically removed while preserving logical validity. The final form is a disjunction of chains without $\del$, which evaluates to 1. \\
        This example also illustrates the fact that the non-standard semantics is not closed under substitution. Assume we substitute $\top$ for $a$, we obtain 
\[
\top\lor \del (\top\lor \top \limp \bot) \equiv \neg \top
\]
and consequently $\top\equiv \bot$.
    \end{remark}

\begin{remark}
Expressions of the form
$\bot < a < c$
are to be understood as abbreviations for conjunctions of pairwise inequalities. More precisely, we define:
\[
\bot < a < c \;:=\; (\bot < a) \land (a < c).
\]
Similarly, longer chains are interpreted conjunctively. For example:
\[
\bot < a < b < \top \;:=\; (\bot < a) \land (a < b) \land (b < \top).
\]
Without this explicit convention, such expressions would be ambiguous.
\end{remark}

\subsection{Standard VS. Non-standard Semantics}

In this section, we establish some facts about standard and non-standard semantics at the propositional level considering the behavior of their corresponding chains.

The proposition below says that standard and non-standard semantics coincide at the propositional level without $\del$. The reason it works is that, without $\del$, formulas cannot distinguish between “exactly 1” and “arbitrarily close to 1” in the way $\del$ can. So if a counterexample uses some variables at value 1, one can modify the corresponding chain description and obtain a counterexample in the non-standard setting.

\begin{proposition}\label{equivalence}
Let $A$ be a propositional formula in the language without $\Delta$.
Then $G_{[0,1]}\models A \Leftrightarrow G_{[0,1]}^*\models A.$
\end{proposition}

\begin{proof}
One direction is trivial: if $A$ is valid in the ordinary semantics, then it is valid in the non-standard semantics, since the non-standard semantics allows fewer valuations and thus chains.

For the converse, assume that $A$ is not valid in the standard semantics. 
We have to consider the following cases:
\begin{itemize}
    \item If $A$ is $\bot$ it is trivial.
    \item If $A$ is $\top$ then it is not a counterexample.
    \item Then there is some chain $C$ and a variable $X_i$ which is not in the equivalence class of $\top$ such that
\end{itemize}
\[C\land A \equiv C\land X_i\]
and is valid in $G_{[0,1]}$. 
Suppose that some variables $\{Y_1\dots Y_n\}$ in the chain $C$ are equivalent to $\top$
\[C:= \bot \vartriangleright_1 X_{\Pi(1)}\land  \dots \land
 X_{\Pi(k-1)}\vartriangleright_k Y_1\leftrightarrow Y_2 \dots \leftrightarrow Y_n\leftrightarrow \top.\]
We now modify $C$ to
\[C^*:= \bot \vartriangleright_1 X_{\Pi(1)}\land  \dots \land
 X_{\Pi(k-1)}\vartriangleright_k  Y_1\leftrightarrow Y_2 \dots \leftrightarrow Y_n < \top.\]
 The evaluation is again of the form $C^* \land X_i$. This leads to a counterexample in non-standard semantics as $C^*$ is a non-standard chain. 
\end{proof}
We have the following immediate corollaries:

\begin{corollary}
    Let $A$ be a $\del$-free propositional formula.
    \begin{enumerate}
        \item Then $A$ is valid in Gödel logic if and only if it evaluates to $1$ under all non-standard chains. In particular, in chain normal form it suffices to consider only non-standard chains. 
        \item If $A$ is valid in the non-standard semantics, then any formula obtained from $A$ by replacing propositional variables by $\top$ is also valid in the non-standard semantics.
    \end{enumerate}
\end{corollary}

\begin{proof}
 \begin{enumerate}
     \item By the above Proposition \ref{equivalence}, validity in the standard semantics coincides with validity in the non-standard semantics. Hence, it suffices to check validity over non-standard chains only.
     \item If $A$ is valid, then by the above Proposition \ref{equivalence}, it is valid in the standard and non-standard semantics. In standard semantics, assigning to a variable the value $1$ corresponds to replacing it with $\top$. Since validity is preserved under such substitutions, the resulting formula remains valid in both standard and non-standard semantics.
 \end{enumerate}   
\end{proof}

\begin{remark}
    Note that Gödel logics in non-standard semantics generally do not allow substitution-closure, therefore, the replacement by $\top$ is not trivial.
\end{remark}

\begin{corollary}
The equivalence between standard and non-standard semantics fails in general for formulas containing $\del$. In particular, there exists formula $A$ with $\del$ such that $A$ is valid in non-standard semantics but not valid in standard semantics.
\end{corollary}

\begin{proof}
The connective $\del$ distinguishes between value $1$ and values strictly below $1$. Hence the transformation used in the proof of the Proposition \ref{equivalence} is no longer sound in the presence of $\del$, and the converse implication fails in general as $\neg\del X$ where $X$ is a variable valid in the non-standard semantics but not in the standard semantics.
\end{proof}


\section{First-Order Extension: Witnessed Gödel Logics}
In the first-order setting, the behavior of the $\del$-operator becomes significantly more subtle. While the propositional case admits elimination under non-standard semantics, the first-order case introduces additional complications related to quantification and the structure of truth value sets. Thus, the argument used for the propositional case does not extend straightforward to the first-order case.
 when 1 is not isolated and not in a perfect set, however, 0 is isolated or in a perfect set, the first-order Gödel logic with $\del$ is not recursively enumerable, while without $\del$ it is. This holds both for standard and non-standard semantics. Therefore, there is not even an effective-validity equivalence elimination of $\del$, and, obviously, no valid equivalence as in the propositional case.

\begin{theorem}[cf. \cite{Baaz-Preining:Gödel–Dummett}]
A first-order G\"odel logic \(G_V\) is recursively enumerable iff one of the following conditions is satisfied:\\
1. \(V\) is finite,\\
2. \(V\) is uncountable 
and $0$ is an isolated point,\\
3. \(V\) is uncountable, and $0$ is in the prefect subset. 
\end{theorem}
\begin{theorem}
    [cf. \cite{Baaz-Preining-Zach:FirstOrderGoedel}]
A first-order G\"odel logic \(G_V^\del\) is recursively enumerable iff one of the following conditions is satisfied:\\
1. \(V\) is finite,\\
2. \(V\) is uncountable 
and $0$ is an isolated point, and $1$ is in the perfect subset,\\
3. \(V\) is uncountable, and $0$ is in the prefect subset, and $1$ is isolated point, 
\\
4. \(V\) is uncountable 
and both $0$ and $1$  are isolated points,\\
5. \(V\) is uncountable, and $0$ and $1$ is in the prefect subset. 
\end{theorem}
One can conclude that the non-standard semantics does not behave as the standard semantics. 
 
\begin{corollary}
 If a Gödel logic with or without $\del$ is r.e. in standard semantics then it is also r.e. in the non-standard semantics (but it does not mean that $\del$ is eliminable).
\end{corollary}


\subsection{Standard Semantics: Effective Validity-equivalent $\del$-Elimination}

The $\del$-operator can be eliminated by means of a translation as shown in \cite{BaazFasching2009}.
There exists an effective translation
$A \mapsto A^{T}$
from formulas containing $\del$ to formulas without $\del$ such that
$GW_V^{\del} \models A
\,\, \text{iff} \,\,
GW_V \models A^{T}.$
In other words, the validity in witnessed Gödel logic with $\del$
can be reduced to the validity in witnessed Gödel logic without $\del$.
The translation eliminating the $\del$-operator uses a structural
normalization of formulas.
Every formula $A$ is rewritten in the following form:
\[
\left(\bigwedge_{S} \del \forall x \big(F_S(x) \leftrightarrow S(x)\big)\right)
\rightarrow F_A .
\]
Here the conjunction ranges over all subformulas $S$ of $A$, and
$F_S$ are fresh predicate symbols representing these subformulas.
Thus, each complex subformula is replaced by a new predicate together
with a defining equivalence.
This transformation is similar to the use of definitional
transformations in conjunctive normal form (CNF).
\begin{proposition}\label{prop 7}
The equivalence theorem holds for $G_V$ and $GW_V$ in the form
\[
\forall x\,(A(x) \leftrightarrow B(x)) \rightarrow
\big(C(A(t)) \leftrightarrow C(B(t))\big),
\]
and for $G_V^{\del}$ and $GW_V^{\del}$ in the form
\[
 \forall x \del\,(A(x) \leftrightarrow B(x)) \rightarrow
\big(C(A(t)) \leftrightarrow C(B(t))\big),
\]
where $t$ may contain bound variables and $C$ is a context.
\end{proposition}

\begin{proof}
The proof proceeds by induction on the complexity of $C$. In the
$\del$-case, the proof uses the axioms
$\axdi, \axdiv$
and that validity of $A$ implies validity of $\del A$.
\end{proof}

\begin{definition}
Let $A$ be a formula in the language with $\del$. The
\emph{structural normal form} $\operatorname{struc}(A)$ is obtained by
replacing stepwise all subformulas $S(\bar{x})$ by fresh predicate
symbols $F_S(\bar{x})$. Then
\[
\operatorname{struc}(A) :=
\left(\bigwedge_{S} \del \forall \bar{x}\,
\big(F_S(\bar{x}) \leftrightarrow S(\bar{x})\big)\right)
\rightarrow F_A.
\]
\end{definition}

\begin{proposition}
For all formula $A$,
$G_V^{\del} \models A \,\, \text{iff} \,\,
G_V^{\del} \models \operatorname{struc}(A),$
and similarly,
$GW_V^{\del} \models A \,\, \text{iff} \,\,
GW_V^{\del} \models \operatorname{struc}(A).$
\end{proposition}

\begin{proof}
Replacing $F_S(\bar{x})$ by $S(\bar{x})$ using the equivalence theorem
yields equivalence with the original formula. The argument uses the
$\del$-axioms and rules.
\end{proof}

\begin{theorem}[cf. \cite{BaazFasching2009}]
There exists a translation $T$ from formulas with $\del$ to formulas
without $\del$ such that $GW_V^{\del} \models A
\,\, \text{iff} \,\,
GW_V \models A^{T}.$
\end{theorem}

\begin{proof}
The construction proceeds as follows: write $\operatorname{struc}(A)$ as
\[
\left(
\bigwedge_{H} \del \forall \bar{x}\,
\big(F_{\del H}(\bar{x}) \leftrightarrow \del H(\bar{x})\big)
\;\wedge\;
\bigwedge_{S} \del \forall \bar{x}\,
\big(F_S(\bar{x}) \leftrightarrow S(\bar{x})\big)
\right)
\rightarrow F_A,
\]
where the first conjunction ranges over $\del$-subformulas.

The translation $A^{T}$ replaces the defining equivalences for
$\del$-subformulas by $\del$-free conditions:
\[
\left(
\bigwedge_{H}
\big(
\forall \bar{x}(\neg F_{\del H}(\bar{x}) \vee F_{\del H}(\bar{x}))
\;\wedge\;
\forall \bar{x}(F_{\del H}(\bar{x}) \rightarrow H(\bar{x}))
\big)
\;\wedge\;
\bigwedge_{S}
\forall \bar{x}(F_S(\bar{x}) \leftrightarrow S(\bar{x}))
\right)
\rightarrow
\]
\[\rightarrow\left(
F_A \;\vee\;
\bigvee_{H} \exists \bar{x}
\big(\neg F_{\del H}(\bar{x}) \wedge H(\bar{x})\big)
\right).\]
One then proves:
\begin{itemize}
\item If $GW_V \models A^{T}$, then $GW_V^{\del} \models A$ by
reintroducing $\del$ and using the $\del$-axioms.
\item Conversely, if $A^{T}$ is not valid in $GW_V$, a witnessed
countermodel can be constructed showing that $A$ is not valid in
$GW_V^{\del}$.
\end{itemize}
Thus the equivalence holds.
\end{proof}
The main result of this section is that the $\del$-operator can be
eliminated syntactically via translation. This reduces reasoning in
$GW_V^{\del}$ to reasoning in $GW_V$, which is particularly useful
for proof theory and automated reasoning. Note that this translation works also in standard semantics but fails in non-standard semantics.

\subsection{
Non-standard Semantics: Logically-Equivalent $\del$-Elimination}

 In the witnessed setting, a full collection of quantifier shifting principles is available and axiomatized. These allow us to move all quantifiers in prenex form and eliminate occurrences of \( \del \) from the quantifier-free matrix of a formula via elementary equivalences discussed in Chapter \ref{propositional}.
As a result, \( \del \)-elimination and related completeness arguments remain valid in witnessed Gödel logics. 
Note that this method only works in the non-standard semantics if we allow witnessed settings but fails in standard semantics. 

As a consequence of Remark~\ref{remark13}, in first-order Gödel logics, $\del$ is eliminable in the non-standard semantics under the assumption of witnessed interpretations, and in the standard semantics for the prenex fragments. The status of $\del$-elimination in general first-order setting, beyond finite, witnessed, and prenex cases, remains open.

\section{Conclusion}

We conjecture that in the non-standard semantics all uncountable infinite-valued $G_V^*$ coincide with $G_{[0,1]}^*$ and all countable infinite-valued $GW_V^*$ are not r.e. This means that in Gödel logic $G_V$ where $1$ is not isolated, not in perfect set and $0$ is in the perfect set, there is no effective validity-equivalence translation between $G^*_\del$ and $G^*$.

We conjecture that a logically-equivalent elimination of $\del$ does not hold for $G_{[0,1]}^*$ for the usual first-order setting.

Additionally, this framework suggests broader implications for other intermediate and modal logics, opening new possibilities for research into similar semantic constraints and their logical impacts.

Another interesting question for further investigation could be
an embedding of the results into the algebraic framework, especially in connection with Heyting algebras.



%
%
%
%
\end{document}